\documentclass{article}

\usepackage{soul,graphicx,amsmath,amsthm,a4wide,hyperref}

\newtheorem{theorem}{\bf Theorem}[section]
\newtheorem{lemma}{\bf Lemma}[section]
\newtheorem{prop}{\bf Proposition}[section]

\begin{document}

\title{Revisiting the exclusion principle in epidemiology at its ultimate limit}

\author{Nir Gavish\footnote{Faculty of Mathematics, Technion, Haifa, 32000, Israel; \url{ngavish@technion.ac.il}}}
\date{}
\maketitle

\begin{abstract}
The competitive exclusion principle in epidemiology implies that when competing strains of a pathogen provide complete protection for each other, the strain with the largest reproduction number outcompetes the other strains and drives them to extinction.  The introduction of various trade-off mechanisms may facilitate the coexistence of competing strains, especially when their respective basic reproduction numbers are close so that the competition between the strains is weak.  Yet, one may expect that a substantial competitive advantage of one of the strains will eventually outbalance trade-off mechanisms driving less competitive strains to extinction.  The literature, however, lacks a rigorous validation of this statement.\\
In this work, we challenge the validity of the exclusion principle at an ultimate limit in which one strain has a vast competitive advantage over the other strains.  We show that when one strain is significantly more transmissible than the others, and under broad conditions, an epidemic system with two strains has a stable endemic equilibrium in which both strains coexist with comparable prevalence.  Thus, the competitive exclusion principle does not unconditionally hold beyond the established case of complete immunity.  
\end{abstract}
\section{Introduction}

One of the important principles of theoretical biology is the {\em competitive exclusion principle} which states that no two species can indefinitely occupy the same ecological niche~\cite{levin1970community,doi:10.1126/science.131.3409.1292,gause1934experimental,volterra1927variazioni}. Many epidemiological systems, e.g., Dengue, Influenza or Malaria, consist of multiple strains of a pathogen which all compete for susceptibles, see~\cite{martcheva2015introduction,Keeling-book} and references within.   The competitive exclusion principle in epidemiology implies that when competing strains of a pathogen provide complete protection for each other, the strain with the largest reproduction number outcompetes the other strains and drives them to extinction~\cite{bremermann1989competitive}.  The result of~\cite{bremermann1989competitive} applies to directly transmitted infections with cross-immunity in a homogeneous population. Subsequent works showed that the exclusion principle extends to a broad family of epidemic systems including those described by vector-borne disease models~\cite{feng1997competitive,zheng2020competitive,cai2013competitive}, age-structured models~\cite{martcheva2013competitive}, network structures~\cite{wang2022competitive}, and also to models incorporating distributed delays and environmental transmission~\cite{cai2013competitive,dang2016competitive}.

In various cases, e.g., Dengue or Malaria, multiple strains of pathogens of infectious diseases stably coexist for long periods despite competing with each other, see~\cite{martcheva2015introduction,Keeling-book} and references within. 
Extensive research efforts have been dedicated to understanding how such competing species coexist rather than following the competitive exclusion principle and driving less competitive strains to extinction.  These studies showed coexistence is possible when strain interactions and model dynamics extends beyond complete cross-immunity.  Particularly, these studies revealed
numerous trade-off mechanisms that may enable the coexistence of pathogen variants~\cite{martcheva2008vaccine,thieme2007pathogen}.


{\em Partial cross-immunity}, for example, is a trade-off mechanism generating an immune response after infection that provides partial protection against subsequent infections with related strains, but not complete protection. As a result, competition is reduced, facilitating coexistence between strains~\cite{chung2016dynamics}. Similarly, super-infection and co-infection reduce competition and facilitate coexistence by enabling an infected individual to be infected by other strains before recovery~\cite{martcheva2013linking,candela2009coinfection}. Other examples of trade-off mechanisms involve adaptation by mutation, quarantine or age structure~\cite{nuno2005dynamics,nuno2008mathematical,thieme2007pathogen,kuddus2022analysis}. 

The above studies of trade-off mechanisms, however, consider cases in which the basic reproduction numbers of the two strains are comparable, and, therefore, the competition between the two strains is relatively weak.  
For example, the study of a multi-strain system in~\cite{white1998cross} shows that the coexistence of strains is possible when their {\em respective basic reproduction numbers are close and cross-immunity is weak}. 
Similarly, a series of works shows how strain structure and partial cross-immunity between {\em strains with close respective basic reproduction numbers} can lead to long-period oscillatory dynamics without external forcing, see~\cite[Sec. 4.1.4.2.]{Keeling-book} and references within.  The arising picture is that  weaker competition implies weaker competitive exclusion, allowing various trade-off mechanisms to dominate dynamics. 

Intuitively, the competitive exclusion principle should be most valid and robust in the limit of fierce competition between the strains.  
Thus, if one of the strains has a significantly larger competitive advantage over the other, one may expect it will outbalance trade-off mechanisms and become the dominant and eventually the sole strain. However, to the best of our knowledge, the literature lacks a rigorous validation of this statement. This is the focus of this study.

In this work, we test the validity of the exclusion principle in an ultimate regime in which one strain has a vast competitive advantage over the other strains, in the sense that one of the strains is significantly more transmissible than the other strains.   
We address the following questions:\\
{\bf 1)} Does the competitive exclusion principle always hold when one strain has a vast competitive advantage over the other strains? Namely, beyond the established case of complete cross-immunity, would a sufficiently large competitive advantage of one of the strain outbalance trade-off mechanisms?

Our findings are counter-intuitive and show that in a very general setting, the competitive exclusion principle does not hold, unconditionally, beyond the case of complete cross-immunity. Rather, we show that such systems may have a stable endemic equilibrium in which both species coexist.  This raises further questions,

{\bf 2)} What trade-off mechanisms contribute to coexistence, even when one strain has a vast competitive advantage over the others?

Our results reveal the key role of partial cross-immunity in enabling coexistence.

\section{Methods}
\subsection{Transmission model}\label{subsec:transmission}
\begin{figure}[ht!]
\begin{center}
\includegraphics[width=0.55\textwidth]{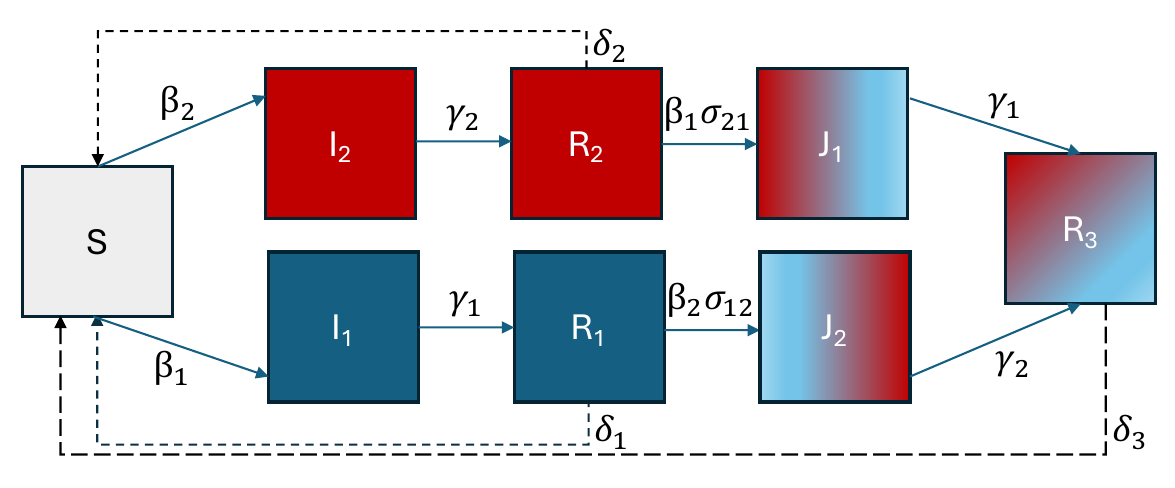}
\caption{Schematic diagram of disease dynamics when the host is exposed to two co-circulating strains. Susceptibles (S) may be infected with strain~$i=1$ or~$i=2$ ($I_i$, primary infection).  Those recovered from strain~$i$ ($R_i$, as a result of primary infection) are temporarily immune to reinfections by strain~$i$ but are possibly susceptible to infections by strain~$j\ne i$ ($J_j$, secondary infection).  Here,~$\sigma_{ij}$ is the relative susceptibility to strain~$j$ for an individual previously infected with and recovered from strain~$i$ ($i\ne j$), so that~$\sigma_{ij}=0$ corresponds to total cross-immunity,~$0<\sigma_{ij}<1$ corresponds to partial cross-immunity and~$\sigma_{ij}>1$ corresponds to enhanced susceptibility.}
\label{fig:diagram}
\end{center}
\end{figure} 

 We consider a two-strain epidemic model in which the strains can interact indirectly via the immune response generated following infections, and in which this immune response possibly wanes with time, see diagram in Figure~\ref{fig:diagram}.
The mathematical model is based on the model presented in~\cite{castillo1989epidemiological,martcheva2015introduction}, and extended to account for waning immunity and possible enhanced or reduced infectivity following an infection.  Each individual resides in one of eight compartments: susceptibles (S), those infected with strain~$i$ ($I_i$, primary infection), those infected with strain~$i$ ($J_i$, secondary infection), after having been previously infected (and recovered) by strain~$j\ne i$, those recovered from strain~$i$ ($R_i$, as a result of primary infection), and those recovered from both strains ($R_3$). The population is assumed to mix randomly.  The model is given by
\begin{subequations}\label{eq:model} 
\begin{align}
        \frac{\text{d}S}{\text{d}t} &=\mu -\sum_{i=1}^{2}\beta_i(I_i+\eta_iJ_i)S-\mu  S+\sum_{k=1}^{3}\delta_k  R_k,\label{eq:S}\\[1ex]
         \frac{\text{d}I_1}{\text{d}t}  &= \beta_1S(I_1+\eta_{1}J_1)-(\mu +\gamma_1)I_1,\label{eq:eqI1}\\[1ex]
         \frac{\text{d}I_2}{\text{d}t}  &= \beta_2S(I_2+\eta_{2}J_2)-(\mu +\gamma_2)I_2,\label{eq:eqI2}\\[1ex]
         \frac{\text{d}J_1}{\text{d}t}  &= \beta_1\sigma_{21}R_2(I_1+\eta_{1}J_1)-(\mu +\gamma_1)J_1,\label{eq:eqJ1}\\[1ex]
         \frac{\text{d}J_2}{\text{d}t}  &= \beta_2\sigma_{12}R_1(I_2+\eta_{2}J_2)-(\mu +\gamma_2)J_2,\label{eq:eqJ2}\\[1ex]
         \frac{\text{d}{R_1}}{\text{d}t} &= \gamma_1I_1-\beta_2\sigma_{12}(I_2+\eta_{2}J_2)R_1-(\mu +\delta_1)R_1,\label{eq:eqR1}\\[1ex]
         \frac{\text{d}{R_2}}{\text{d}t} &= \gamma_2I_2-\beta_1\sigma_{21}(I_1+\eta_{1}J_1)R_2-(\mu +\delta_2)R_2,\label{eq:eqR2}\\[1ex]
               \frac{\text{d}R_3}{\text{d}t} &= \gamma_1J_1+\gamma_2J_2-(\mu  +\delta_3)R_3,
        \end{align}
\end{subequations}
were~$\mu$ is the rate at which individuals are born, as well as the mortality rate,~$\beta_i>0$ denotes the transmission coefficient for strain~$i$,~$\gamma_i$ denotes the recovery rate from strain~$i$,~$\delta_i$ is the rate at which immunity to re-infection by strain~$i$ wanes and~$\delta_3$ is the rate at which immunity to re-infection following infection by both strains wanes.  The relative infectivity of those infected with strain~$i$ after previous infection with strain~$j\ne i$ is given by~$\eta_i$.  Finally,~$\sigma_{ij}$ is the relative susceptibility to strain~$j$ for an individual previously infected with and recovered from strain~$i$ ($i\ne j$), so that~$\sigma_{ij}=0$ corresponds to total cross-immunity,~$0<\sigma_{ij}<1$ corresponds to partial cross-immunity and~$\sigma_{ij}>1$ corresponds to enhanced susceptibility. 

\subsubsection{Assumptions on model parameters}
The feasible range of the problem parameters is
\begin{subequations}\label{eq:condParms}
\begin{equation}
\beta_i>0,\quad \gamma_i>0,\quad \mu\ge0,\quad \delta_k\ge0,\quad
\eta_i> 0,\quad \sigma_{ij}\ge0,\qquad i=1,2, \quad  j\ne i \quad k=1,2,3.
\end{equation}
In what follows, we assume that the susceptible group is replenished by demographic turnover ($\mu>0$) and/or by waning of the immune response generated following infections~($\delta_k>0$),
\begin{equation}
\max\{\mu,\delta_1,\delta_2,\delta_3\}>0.
\end{equation}
\end{subequations}
This condition enables the system to converge to an endemic equilibrium, rather than gradually exhausting the susceptible pool and converging to a disease-free state.

\subsubsection{Basic reproduction number}\label{sec:R0}
A standard computation of the next generation matrix~\cite{martcheva2015introduction,diekmann1990definition,van2002reproduction} yields that the basic reproduction number equals
\[
\mathcal{R}_0=\max\{\mathcal{R}_1,\mathcal{R}_2\}
\]
where~$\mathcal{R}_i$ is the basic reproduction number of strain~$i$
\begin{equation}\label{eq:Ri}
\mathcal{R}_i=\frac{\beta_i}{\mu +\gamma_i}.
\end{equation}
See Appendix~\ref{app:R0} for details.

\subsection{Analysis in the case of a highly transmissible strain}\label{subsec:phiCE}
We focus on the case in which strain one has a vast competitive advantage over strain two, and ask whether, in such a case, the system can give rise to long-term dynamics in which both strains coexist.  In particular, we consider the case in which strain one is significantly more transmissible than strain two,~$\beta_1\gg 1$,
so that in terms of its basic reproduction number,
\begin{equation}\label{eq:asymptotic_regime_I}
\mathcal{R}_1\gg1.
\end{equation}

The following Theorem characterizes single-strain endemic equilibrium points of~\eqref{eq:model} in the asymptotic regime~\eqref{eq:asymptotic_regime_I}.
\begin{theorem}\label{thm:EE1}
Let~$\mu$, $\{\gamma_i\}_{i=1}^2$, $\{\delta_k\}_{k=1}^3$, $\{\eta_i\}_{i=1}^2$,~$\sigma_{12}$ and~$\sigma_{21}$ 
be parameters that satisfy~\eqref{eq:condParms}.  Then, for sufficiently large~$\beta_1>\gamma_1$,
\begin{enumerate}
\item 
The system~\eqref{eq:model} has a unique single-strain endemic equilibrium~$\phi^{EE,1}$ with~$I_1>0$ and~$I_2=0$ that satisfies
\[
\begin{split}
&S^{EE,1}=\frac1{\mathcal{R}_1},\quad I_1^{EE,1}=\frac{\delta_1+\mu }{1+\delta_1+\mu }\frac{\mathcal{R}_1-1}{\mathcal{R}_1},\quad R_1^{EE,1}=\frac{1}{\delta_1+\mu }I_1^{EE,1},
\quad I_2^{EE,1}=J_1^{EE,1}=J_2^{EE,1}=R_2^{EE,1}=0,
\end{split}
\]
where~$\mathcal{R}_i$ is given by~\eqref{eq:Ri}.
\item The solution~$\phi^{EE,1}$ is locally unstable when
\begin{equation}\label{eq:IEE1unstable}
\hat{\mathcal{R}}^1_2:=\frac{\sigma_{12}\gamma_1\eta_2}{\gamma_1+\delta_1+\mu }\mathcal{R}_2+\left[1-\frac{\sigma_{12}\gamma_1\eta_2}{\gamma_1+\delta_1+\mu }\right]\frac{\mathcal{R}_2}{\mathcal{R}_1}>1,
\end{equation}
and locally stable when~$\hat{\mathcal{R}}^1_2<1$.
\item If~$\mathcal{R}_2>1$, the system~\eqref{eq:model} has an additional unique single-strain endemic equilibrium~$\phi^{EE,2}$ with~$I_2>0$ and~$I_1=0$ which satisfies
\[
\begin{split}
&S^{EE,2}=\frac1{\mathcal{R}_2},\quad I_2^{EE,2}=\frac{\delta_2+\mu }{1+\delta_2+\mu }\frac{\mathcal{R}_2-1}{\mathcal{R}_2},\quad R_2^{EE,2}=\frac{1}{\delta_2+\mu }I_1^{EE,2},
\quad I_1^{EE,2}=J_1^{EE,2}=J_2^{EE,2}=R_1^{EE,2}=0.
\end{split}
\]
This solution is unstable.
\end{enumerate}
\end{theorem}
\begin{proof}
See Appendix~\ref{app:proofLemma21}.
\end{proof}
Theorem~\ref{thm:EE1} implies that when~$\mathcal{R}_1\gg1$ and the invasion number,~$\hat{\mathcal{R}}^1_2$, of strain two is bigger than one, then both single-strain endemic equilibrium points, if exist, are unstable.  The disease-free equilibrium is also unstable in the regime~$\mathcal{R}_1\gg1$, see Section~\ref{sec:R0}.  Thus, the system may potentially give rise to coexistence in the regime~$\mathcal{R}_1\gg1$ and~$\hat{\mathcal{R}}^1_2>1$.  However, instability of the disease-free equilibrium and the single-strain endemic equilibria does not assure the persistence of the two strains~\cite{margheri2003some}.  Moreover, if system~\eqref{eq:model} does give rise to coexisting strains in the long run, Theorem~\ref{thm:EE1} does not provide any information on the nature of coexistence, e.g., does the system oscillate between infection waves of different strains, converge to an endemic equilibrium or develop chaos.   

The following theorem complements the results of Theorem~\ref{thm:EE1} by characterizing cases in which the system of study has a stable endemic equilibrium in which both strains coexist, despite the vast competitive advantage of strain one:
\begin{theorem}\label{thm:phi*stability}
Let~$\beta_2$,~$\mu$, $\{\gamma_i\}_{i=1}^2$, $\{\delta_k\}_{k=1}^3$, $\{\eta_i\}_{i=1}^2$,~$\sigma_{12}$ and~$\sigma_{21}$ 
be parameters that satisfy~\eqref{eq:condParms}.
If,~$\sigma_{12}>0$ and \begin{equation}\label{eq:beta2cond}
\mathcal{R}_2>\frac{\gamma_1+\delta_1+\mu}{\gamma_1\eta_2\sigma_{12}},
\end{equation}
where~$\mathcal{R}_i$ is given by~\eqref{eq:Ri},
then for sufficiently large~$\beta_1>\gamma_1$ 
\\1) There exists a unique steady-state solution $$\phi^*=(S^*,I_1^*,I_2^*,J_1^*,J_2^*,R_1^*,R_2^*,R_3^*)$$ of the system~\eqref{eq:model} with~$I_1^*>0$ and~$I_2^*>0$ that satisfies
\begin{subequations}\label{eq:approxCE}
\begin{equation}\label{eq:series_S}
S^*=\frac1{\mathcal{R}_1}-\frac{s_2}{\mathcal{R}_1^2}+\mathcal{O}\left(\frac1{\mathcal{R}_1^3}\right),\quad 
\quad I_1^*=a_1-\frac{b_1}{\mathcal{R}_1} +\mathcal{O}\left(\frac1{\mathcal{R}_1^2}\right),\quad I_2^*=\frac{b_2}{\mathcal{R}_1}+\mathcal{O}\left(\frac1{\mathcal{R}_1^2}\right),
\end{equation}
\begin{equation}\label{eq:J2R1_func_S_Ii}
J_2^* = \frac{\gamma_1}{\gamma_2+\mu }I_1^* -\frac{\mu +\delta_1}{\gamma_2+\mu }R_1^*=\frac{\gamma_1a_1}{\gamma_2+\mu}-\frac{\delta_1+\mu}{\eta_2(\gamma_2+\mu)\sigma_{12}\mathcal{R}_2}+\mathcal{O}\left(\frac1{\mathcal{R}_1}\right),\qquad R^*_1=\frac{1-\mathcal{R}_2S^*}{\eta_2\mathcal{R}_2\sigma_{12}}=\frac{1}{\eta_2\sigma_{12}\mathcal{R}_2}+\mathcal{O}\left(\frac1{\mathcal{R}_1}\right)
\end{equation}
\begin{equation}\label{eq:b2a1}
b_2=\frac{\gamma_1\eta_2(\delta_3+\mu)}{\gamma_1\gamma_2+(\gamma_1+\gamma_2+\mu)(\delta_3+\mu)}\left[\mathcal{R}_2-\frac{\gamma_1+\delta_1+\mu}{\gamma_1\eta_2\sigma_{12}}\right],\quad 
a_1=\frac{(\gamma_2+\mu)\sigma_{12}b_2+\delta_1+\mu}{\gamma_1\eta_2\sigma_{12}\mathcal{R}_2},
\end{equation}
where for~$\sigma_{21}>0$
\begin{equation}\label{eq:J1R2_func_S_Ii_sigma21_positive}
J_1^*=\frac{\gamma_2}{\gamma_1+\mu } I_2^* -\frac{\mu +\delta_2}{\gamma_1+\mu }R_2^*=\frac{\gamma_2b_2}{\gamma_1+\mu}\frac1{\mathcal{R}_1}+\mathcal{O}\left(\frac1{\mathcal{R}_1^2}\right),\qquad R_2^*=\frac{1-\mathcal{R}_1S^*}{\eta_1\mathcal{R}_1\sigma_{21}}=\frac{\gamma_2b_2}{a_1\sigma_{21}(\gamma_1+\mu)}\frac1{\mathcal{R}_1^2}+\mathcal{O}\left(\frac1{\mathcal{R}_1^3}\right),
\end{equation}
\begin{equation}\label{eq:s2b1_sigma21_positive}
s_2=\frac{\gamma_2\eta_1b_2}{(\mu + \gamma_1)a_1},\quad
b_1=\frac{\gamma_2+\mu}{\gamma_1+\mu}b_2+\frac{\sigma_{12}(\gamma_2+\mu)(\delta_3+\mu)+(\gamma_2+\delta_3+\mu)\delta_1-\gamma_2\delta_3}{(\gamma_1+\gamma_2)\delta_3+\gamma_1\gamma_2+(\gamma_1+\gamma_2+\delta_3)\mu+\mu^3}\frac1{\sigma_{12}},
\end{equation}
and for~$\sigma_{21}=0$,
\begin{equation}\label{eq:J1R2_func_S_Ii_sigma21_zero}
J_1^*=0,\qquad R_2^*=\frac{\gamma_2}{\delta_2+\mu}I_2^*=\frac{\gamma_2b_2}{\delta_2+\mu}\frac1{\mathcal{R}_1}+\mathcal{O}\left(\frac1{\mathcal{R}_1^2}\right),
\end{equation}
\begin{equation}\label{eq:s2b1_sigma21_zero}
s_2=0,\quad
b_1=\frac{(\gamma_2+\mu)(\delta_3+\mu)(\gamma_2+\delta_2+\mu)\eta_2\sigma_{12}b_2+(\mu+\delta_2)\left[(\gamma_2+\mu)(\delta_3+\mu)\eta_2\sigma_{12}+\gamma_2(\delta_1-\delta_3)+\delta_1(\delta_3+\mu)\right]}{[(\gamma_1+\gamma_2)\delta_3+\gamma_1\gamma_2+(\gamma_1+\gamma_2+\delta_3)\mu+\mu^3](\delta_2+\mu)\eta_2\sigma_{12}}.
\end{equation}
\end{subequations}
\,\\2) The solution~$\phi^*$ is locally stable.
\end{theorem}
\begin{proof}
See appendix~\ref{app:TheoremCE}
\end{proof}
Theorem~\ref{thm:phi*stability} considers cases when strain one is highly transmissible, and strain two is transmissible enough to satisfy~\eqref{eq:beta2cond}. In these cases, the theorem proves the existence, uniqueness and stability of an endemic equilibrium in which both strains coexist.   Condition~\eqref{eq:beta2cond} identifies with condition~\eqref{eq:IEE1unstable} in the limit~$\beta_1\to\infty$ and can be read as
\[
\lim_{\beta_1\to\infty}\hat{\mathcal{R}}^1_2>1,
\]
i.e., the invasion number of strain two is bigger than one when strain one is highly transmissible.
We note that condition~\eqref{eq:beta2cond} is stronger than condition~\eqref{eq:IEE1unstable}.  Namely, the coexistence endemic equilibrium~$\phi^*$ does not exist when the single-strain endemic equilibrium~$I_{EE}^1$ is stable.
Furthermore, the theorem provides an approximation to the variable values at the endemic equilibrium, see~\eqref{eq:approxCE}. 
Condition~\eqref{eq:beta2cond}, for example, originates from approximation~\eqref{eq:approxCE}.  Indeed, the requirement~$I_2^*>0$ implies~$b_2>0$, which in turn leads to condition~\eqref{eq:beta2cond}, see~\eqref{eq:series_S} and~\eqref{eq:b2a1}.
Approximation~\eqref{eq:approxCE} also implies that despite the vast competitive advantage of strain one, at the equilibrium of coexistence, the prevalence of both strains is comparable
\[
\frac{I_2^*+J_2^*}{I_1^*+J_1^*}=
\frac{\gamma_1\sigma_{12}b_2}{(\gamma_2+\mu)\sigma_{12}b_2+\delta_1+\mu}+\mathcal{O}\left(\frac1{\mathcal{R}_1}\right),
\]
and the prevalence of strain two may exceed that of strain one.  

\section{Results}
Theorem~\ref{thm:phi*stability} implies that two strains can coexist despite fierce competition between the strains.  In the system of study, a necessary condition for coexistence is that infection by the dominant strain does not provide complete protection from infections by the other strain.  Indeed, in the limit of complete cross-immunity,~$\sigma_{12}=0$, condition~\eqref{eq:beta2cond} cannot be satisfied.  Therefore, Theorem~\ref{thm:phi*stability} implies that the competitive exclusion principle does not unconditionally hold beyond the established case of complete cross-immunity~\cite{bremermann1989competitive}, even if one strain has a vast competitive advantage over the others.  

In what follows we present two numerical examples. In both cases, we fix all parameters except the transmission parameter of strain one and study the system's behavior as strain one becomes more and more transmissible.  We choose the timescales so that they roughly correspond to Influenza or COVID-19, where one unit of time is a week: A typical infectious period of one week ($\gamma_i\approx 1$), characteristic convalescent immunity period of a year ($\delta_i\approx 1/50$)~\cite{zhao2018individual,doi:10.1056/NEJMoa2118946}, and lifespan of 80 years ($\mu\approx 1/4000$).
In both examples, the parameters are chosen so that condition~\eqref{eq:beta2cond} of Theorem~\ref{thm:phi*stability} is satisfied.  Hence, we expect the systems to converge to a multi-strain endemic equilibrium when strain one is sufficiently transmissible.

\begin{figure}[ht!]
\begin{center}
\includegraphics[width=0.8\textwidth]{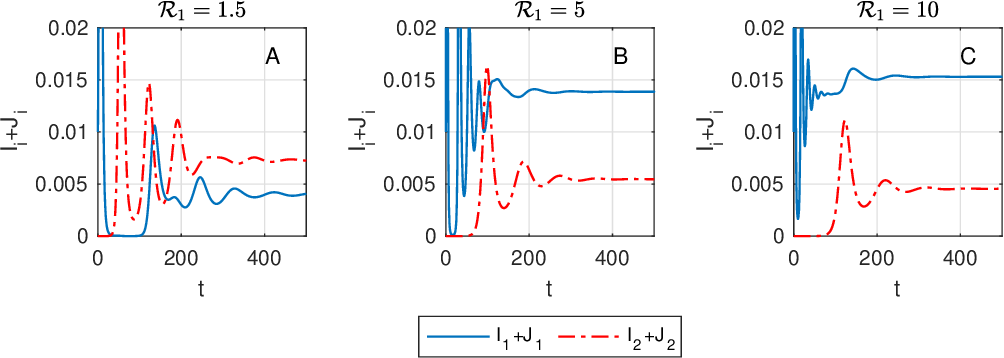}
\caption{Solutions~$I_1(t)$ (solid blue) and~$I_2(t)$ (dash-dotted red) of~\eqref{eq:model} with~$\gamma_1=1.1$, $\gamma_2=0.9$, $\eta_1=0.8$, $\eta_2=1$, $\delta_1=1/60$, $\delta_2=1/50$, $\delta_3=1/30$, $\mu=1/4000$, $\sigma_{21}=0.6$ and~$\sigma_{12}  =0.75$.  The transmission parameter of strain two,~$\beta_2\approx 1.35$, is set so that~$\mathcal{R}_2=1.5$, and the transmission parameter of strain one is set so that A:~$\mathcal{R}_1=1.5$ $(\beta_1\approx 1.65)$, B:~$\mathcal{R}_1=5$ $(\beta_1\approx 5.5)$ and C:~$\mathcal{R}_1=10$ $(\beta_1\approx 11)$.  In all three cases, the initial condition is set as~$I_1(0)=10^{-2}$,~$I_2(0)=10^{-8}$,~$S(0)=1-I_1(0)-I_2(0)$ and zero for all other variables.}
\label{fig:example_balanced_invasion}
\end{center}
\end{figure}We first consider an example in which the strain properties, excluding transmissibility, are comparable with only slight differences in their values.   We chose the parameters so that they correspond to partial cross-immunity~$0<\sigma_{ij}<1$ and neutral or reduced relative infectivity~$0<\eta_i\le1$ in secondary infections.  Since these values do not impose any advantage for secondary infections, one may expect that they would not facilitate coexistence in the case of fierce competition between the strains.  We chose the initial condition~$I_1(0)=10^{-2}$ and~$I_2(0)=10^{-8}$ so that the initial prevalence of strain two is significantly smaller than the initial prevalence of strain one.  This choice sets the initial values far from the equilibrium~$\phi^*$ of coexistence and tests for a possible invasion of strain two.  Figure~\ref{fig:example_balanced_invasion}A presents the prevalence of the strains,~$(I_1+J_1)(t)$ and~$(I_2+J_2)(t)$, in the case both strains have the same basic reproduction number~$\mathcal{R}_1=\mathcal{R}_2=1.5$.  We observe that strain two invades and the system reaches an endemic equilibrium in which both strains coexist with a comparable prevalence of the two strains,~$I_1^*+J_1^*\approx 4\cdot 10^{-3}$ and~$I_2^*+J_2^*\approx 7.2\cdot 10^{-3}$.  Next, we consider the same system when strain one is more transmissible.  As can be expected from Theorem~\ref{thm:phi*stability}, we observe that strain two invades and co-exists with strain one when~$\mathcal{R}_1=5$ and~$\mathcal{R}_1=10$, see Figures~\ref{fig:example_balanced_invasion}B and ~\ref{fig:example_balanced_invasion}C, respectively.  Furthermore, the prevalence of the strains is comparable with~$I_1^*+J_1^*\approx 1.5\cdot 10^{-2}$ and~$I_2^*+J_2^*\approx 0.5\cdot 10^{-2}$.

To demonstrate that partial cross-immunity is necessary for coexistence, as follows from Theorem~\ref{thm:phi*stability}, we solve~\eqref{eq:model} in the case of complete cross-immunity,~$\sigma_{12}=0$, and otherwise with the same parameters as in Figure~\ref{fig:example_balanced_invasion}.  As expected, we observe that strain two is driven to extinction for both moderate and large values of~$\mathcal{R}_1$, see Figure~\ref{fig:example_balanced_invasion_sigma12isZero}.  

Theorem~\ref{thm:phi*stability} proves the local stability of~$\phi^*$ under appropriate conditions, and leaves open the question of whether~$\phi^*$ is globally stable. In Figure~\ref{fig:example_balanced_invasion}, the initial condition is set far from~$\phi^*$.  The less competitive strain successfully invades and the system converges to the coexistence equilibrium~$\phi^*$.  These results, as well as results from additional simulations (data not shown), suggest that the equilibrium is globally stable.  

\begin{figure}[ht!]
\begin{center}
\includegraphics[width=0.8\textwidth]{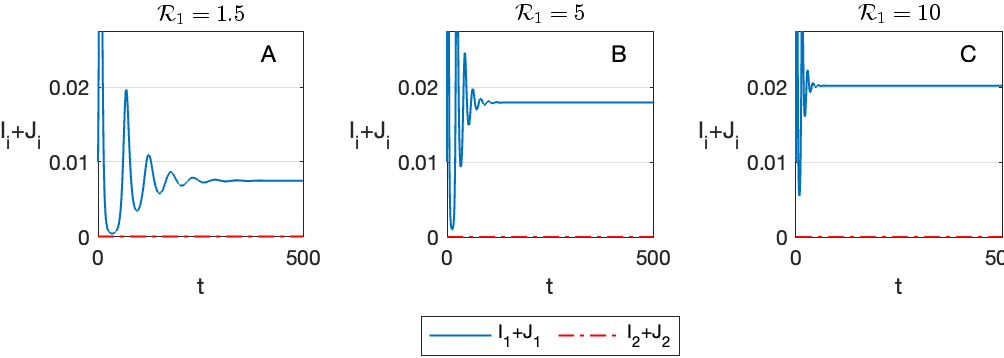}
\caption{Same as Figure~\ref{fig:example_balanced_invasion} but with complete cross-immunity,~$\sigma_{12}=0$.}
\label{fig:example_balanced_invasion_sigma12isZero}
\end{center}
\end{figure}

\begin{figure}[ht!]
\begin{center}
\includegraphics[width=0.8\textwidth]{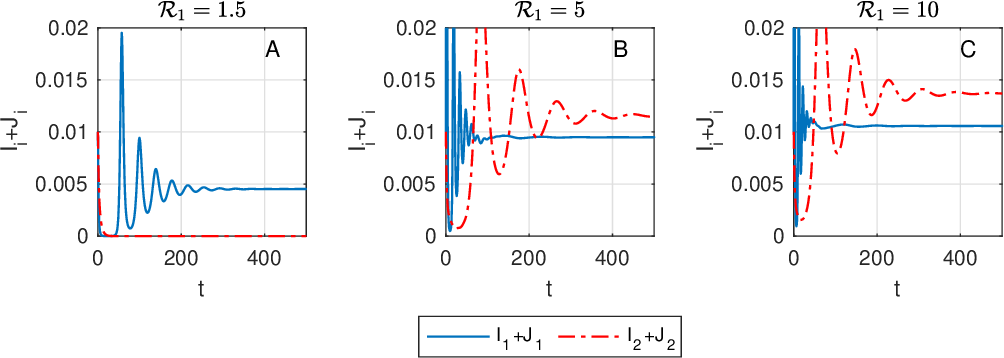}
\caption{Solutions~$I_1(t)$ (solid blue) and~$I_2(t)$ (dash-dotted red) of~\eqref{eq:model} with~$\gamma_1=1.5$, $\gamma_2=0.5$, $\eta_1=1$, $\eta_2=3$, $\delta_1=1/40$, $\delta_2=1/50$, $\delta_3=1/60$, $\mu=1/4000$, $\sigma_{21}=0.75$ and~$\sigma_{12}  =1$.  The transmission parameter of strain two,~$\beta_2\approx 0.3$, is set so that~$\mathcal{R}_2=0.6$, and the transmission parameter of strain one is set so that A:~$\mathcal{R}_1=1.5$ $(\beta_1\approx 2.25)$, B:~$\mathcal{R}_1=5$ $(\beta_1\approx 7.5)$ and C:~$\mathcal{R}_1=10$ $(\beta_1\approx 15)$.  In all three cases, the initial condition is set as~$S(0)=R_1(0)=R_2(0)=0.25$,  $I_1(0)=I_2(0)=10^{-2}$, and zero for all other variables.}
\label{fig:example_StabilizeByLargeR1}
\end{center}
\end{figure}
In the next example, we consider a case in which the basic reproduction number of strain two is smaller than one, yet the relative infectivity in secondary infections with strain two,~$\eta_2=3$, is set relatively high so that condition~\eqref{eq:beta2cond} is satisfied.  In this case, strain two cannot persist alone, yet Theorem~\ref{thm:phi*stability} does imply that it will co-exist with strain one when strain one is sufficiently transmissible.  As expected, for moderate values of~$\mathcal{R}_1$, strain two is driven to extinction, see Figure~\ref{fig:example_StabilizeByLargeR1}A. Yet, for~$\mathcal{R}_1=5$ and~$\mathcal{R}_1=10$, the two strains coexist, see Figures~\ref{fig:example_StabilizeByLargeR1}B and~\ref{fig:example_StabilizeByLargeR1}C, respectively.  Intuitively, strain two benefits from the co-circulation of strain one since the infectivity of those infected with strain two after being infected by strain one is higher than the infectivity of those with primary infections with strain two.  These simulations also provide an example of a case in which the prevalence of the less dominant strain is higher in the long term than the prevalence of the dominant strain, see Figures~\ref{fig:example_StabilizeByLargeR1}B and~\ref{fig:example_StabilizeByLargeR1}C.

A systematic numerical verification of  approximation~\eqref{eq:approxCE} is presented in Appendix~\ref{app:numericalVerification}.

\section{Conclusions}
Our results imply that the competitive exclusion principle does not extend beyond the established case of complete cross-immunity, even when one species has a vast competitive advantage over the other.  Rather, we show that the system can converge to an endemic equilibrium in which both strains coexist, despite fierce competition between the strains.  Particularly, we provide an example of coexistence when the basic reproduction number of the less dominant strain is less than one, and therefore this strain cannot persist without an interaction with the dominant strain.  Moreover, we show that, in equilibrium, the prevalence of the less dominant strain is comparable in magnitude to that of the dominant strain, and even exceeds it in certain cases.

To understand why a vast competitive advantage of strain one does not drive strain two to extinction, it is useful to consider an epidemic outbreak in which both strains are introduced to a susceptible population.  Due to the high basic reproduction number of strain one, in the epidemic's early stages, infections by strain one are dominant and widespread.
While those recovered from strain one remain immune to reinfections by strain one for some time, in the absence of complete cross-immunity, they remain susceptible, to some degree, to infections by strain two.  Strain two can potentially thrive on this susceptible pool without suffering from competition with strain one.  More concretely, strain two will be able to spread if its effective reproduction number~$\mathcal{R}_2^{\rm eff}$ in a population that has all been infected with strain one is larger than one
\[
\mathcal{R}_2^{\rm eff}=\sigma_{12}\mathcal{R}_2\eta_2>1.
\]
Here,~$\mathcal{R}_2^{\rm eff}$ equals the basic reproduction number~$\mathcal{R}_2$ of strain two (for primary infections) while accounting for
the relative changes in susceptibility,~$\sigma_{12}$, and infectivity,~$\eta_2$ of those recovered from strain one.  
This condition is closely related to condition~\eqref{eq:beta2cond} which also takes the effect of waning immunity,~$\delta_1$, mortality,~$\mu$, and the infectious period,~$\gamma_1$, on the size of the population immune to infections by strain one, and can be written as,
\[
\frac{\gamma_1}{\gamma_1+\delta_1+\mu}\mathcal{R}_2^{\rm eff}>1.
\]
Following the explanation above, one may also view strain two as an invading strain after strain one widely spread in the population and reached equilibrium.  Under this view, we note that condition~\eqref{eq:beta2cond} can also be read as the asymptotic invasion number~$\hat{\mathcal{R}^1_2}$ of strain two at the limit of large transmissibility of strain one, see~\eqref{eq:IEE1unstable}.  

The analysis in this work combines two approaches.  First, we follow a common approach and focus on the behavior of single-strain equilibriums while taking advantage of the fact that it is easier to compute the single-strain equilibrium values and determine their local stability, see Theorem~\ref{thm:EE1}.  Yet, the instability of the single-strain and disease-free equilibrium points does not imply the persistence of multiple strains.  Proving persistence may be demanding in complex epidemic systems with interacting strains~\cite{margheri2003some}, and attaining such a result does not reveal the nature of coexistence in the system, e.g., whether the system converges to an endemic equilibrium of coexistence or oscillates.  Therefore, we complement our analysis by studying the possible equilibrium(s) of coexistence.  Since this analysis involves problems that do not have an analytic solution or have a solution that is cumbersome to work with, we use asymptotic methods that enable us to tackle these problems at the accuracy required to attain the solution, e.g., determine the stability of the coexistence state.  Using this approach we show the existence and stability of an endemic equilibrium in which both strains coexist. Therefore, the analysis not only shows coexistence but also characterizes its nature.

Our results show the local stability of an endemic equilibrium of coexistence.  Numerical simulations show that the system converges to the above endemic equilibrium even when the initial conditions are far from equilibrium, suggesting that the equilibrium is globally stable.  Proving global stability is left open as an interesting mathematical question for future research.

One implication of our results is that in the absence of complete cross-immunity, less transmissible, yet more dangerous, strains can invade and coexist with highly transmissible strains.  This suggests that the evolution of pathogens and the emergence of new strains may not maximize~$\mathcal{R}_0$ with time, but rather other quantities.  Such a phenomenon was demonstrated in an epidemic system with superinfection~\cite {nowak1994superinfection}. The study of evolution in a multi-strain system with partial cross-immunity will be presented elsewhere.

Our approach challenged the validity of the exclusion principle in an ultimate limit where, intuitively, it should be most valid.  From a mathematical point of view, studying such a limit enables using approximation methods that allow analysis of the system despite its complexity.  Indeed, many details become negligible at the limit, making the analysis tractable.  Concurrently, other quantities arise naturally at the limit.  Thus, from a biological point of view, such an approach provides valuable insights as it highlights quantities of interest while naturally sorting out details.  The unexpected results revealed using this approach refine our understanding of the exclusion principle and multi-strain dynamics.  

This study was motivated by our experience during the SARS-CoV-2 pandemic.  Indeed, the emergence of the Omicron variant created an unprecedented scenario of an epidemic driven by a highly transmissible variant. Our data-driven research during that time~\cite{gavish2022population} showed that a variant with a basic reproduction number as high as ten can defy conventional theory in certain circumstances.  One example is counter-intuitive optimal vaccine allocations at elevated reproduction numbers~\cite{gavish2022optimal}.  The results presented in this work provide an additional example to an unexpected result that arises when considering highly transmissible strains.  

The counter-intuitive results presented in this study touch upon a fundamental principle in epidemiology and shed light on the mechanisms that enable competing strains to coexist.  While we have focused on the competitive exclusion principle in epidemiology, it is plausible that similar phenomena will also occur in other ecological systems.  For example, microbial populations that evolve near plant roots, and in particular accumulation of host-specific pathogens, may generate negative feedback on the plant and facilitate diversity or coexistence of plant species~\cite{bever2012microbial}, plausibly, even if one plant has a vast competitive advantage over the other.  Extensions of our study beyond epidemiological systems will be presented elsewhere.
\subsection*{Acknowledgements}

This research was supported by the Israel Science Foundation (grant no. 3730/20 to N.G.) within the KillCorona-Curbing Coronavirus Research Program, and by the Israeli Science Foundation (ISF) grant 1596/23.

We thank Guy Katriel for the most helpful comments and discussions on this work.
\appendix

\section{Computation of~$\mathcal{R}_0$}\label{app:R0}

Following a standard computation of the next generation matrix~\cite{martcheva2015introduction,diekmann1990definition,van2002reproduction}
we split the right-hand side in the equations for the infected compartments~$(I_1,I_2,J_1,J_2)$ of~\eqref{eq:model} in the following way:
\begin{equation}\label{eq:NGM}
\begin{split}
\frac{\text{d}I_1}{\text{d}t}  &= \underbrace{\beta_1S(I_1+\eta_{1}J_1)}_{\mathcal{F}_{1}}-\underbrace{(\mu +\gamma_1)I_1}_{\mathcal{V}_1},\\
         \frac{\text{d}I_2}{\text{d}t}  &= \underbrace{\beta_2S(I_2+\eta_{2}J_2)}_{\mathcal{F}_2}-\underbrace{(\mu +\gamma_2)I_2}_{\mathcal{V}_2},\\
         \frac{\text{d}J_1}{\text{d}t}  &= \underbrace{\beta_1\sigma_{21}R_2(I_1+\eta_{1}J_1)}_{\mathcal{F}_3}-\underbrace{(\mu +\gamma_1)J_1}_{\mathcal{V}_3},\\[1ex]
         \frac{\text{d}J_2}{\text{d}t}  &= \underbrace{\beta_2\sigma_{12}R_1(I_2+\eta_{2}J_2)}_{\mathcal{F}_4}-\underbrace{(\mu +\gamma_2)J_2}_{\mathcal{V}_4},
        \end{split}
\end{equation}
        where~$\mathcal{F}_i$ is the appearance rate of new infections in the relevant compartments, and~$\mathcal{V}_i$ incorporates the remaining transitional terms.  Linearization of the system of infected compartments,~\eqref{eq:NGM}, about the disease-free equilibrium in which all the population is susceptible,~$S=1$, yields
\[
x^\prime(t)=(F-V)x,\quad x=(I_1,I_2,J_1,J_2)^T,
\]
where the matrices~$F$ and~$V$ originated from the terms~$\mathcal{F}_i$ and~$\mathcal{V}_i$, respectively,
\[
F=\left[\begin{array}{cccc}
\beta_1&&\beta_1\eta_1&\\
&\beta_2&&\beta_2\eta_2\\
0&0&0&0\\
0&0&0&0
\end{array}
\right],\quad V=\left[\begin{array}{cccc}
\gamma_1+\mu&&&\\
&\gamma_2+\mu&&\\
&&\gamma_1+\mu&\\
&&&\gamma_2+\mu
\end{array}
\right].
\]
The basic reproduction number is the maximal eigenvalue of the next-generation matrix
\[
\mathcal{R}_0=\rho(FV^{-1})=\max\{\mathcal{R}_1,\mathcal{R}_2\},\quad \mathcal{R}_i=\frac{\beta_i}{\gamma_i+\mu}.
\]
\section{Proof of Theorem 2.1}\label{app:proofLemma21}
We first prove part (i) of Theorem~\ref{thm:EE1} by seeking a steady-state solution of~\eqref{eq:model} with~$I_1>0$ and~$I_2=0$.  Substituting~$I_2=0$ in~$J_1^\prime(t)+R_2^\prime(t)=0$, see~\eqref{eq:eqJ1} and~\eqref{eq:eqR2}, implies that~$R_2=0$.  Thus,~$J_1^\prime(t)=0$ implies~$J_1=0$, see~\eqref{eq:eqJ1}. Finally,~$I_2^\prime=0$ implies~$S=0$ or~$J_2=0$, see~\eqref{eq:eqI2}. In the former case, however,~\eqref{eq:eqI1} implies that~$I_1=0$ at steady-state.  Thus, in the case of interest,~$I_1>0$, it follows that~$S\ne0$ and~$J_2=0$.  Overall, the above computation shows that the model~\eqref{eq:model} has a unique single-strain endemic equilibrium point~$\phi^{EE,1}$ with~$I_1>0$ and~$I_2=0$ that satisfies
\[
I_2^{EE,1}=J_1^{EE,1}=J_2^{EE,1}=R_2^{EE,1}=0.
\]
Substituting the above in~\eqref{eq:model} and solving the resulting algebraic system yields that the other variables of~$\phi^{EE,1}$ uniquely satisfy
\[
S^{EE,1}=\frac1{\mathcal{R}_1},\quad I_1^{EE,1}=\frac{\delta_1+\mu }{1+\delta_1+\mu }\frac{\mathcal{R}_1-1}{\mathcal{R}_1},\quad R_1^{EE,1}=\frac{1}{\delta_1+\mu }I_1^{EE,1}.
\]
Note that the condition~$I_1^{EE,1}>0$ implies~$\mathcal{R}_1>1$.  This condition is satisfied for sufficiently large~$\beta_1>\gamma_1$.
Similarly, when~$\mathcal{R}_2>1$, model~\eqref{eq:model} has a unique single-strain endemic steady-state~$\phi^{EE,2}$ with~$I_2>0$ and~$I_1=0$ that satisfies
\[
\begin{split}
&S^{EE,2}=\frac1{\mathcal{R}_2},\quad I_2^{EE,2}=\frac{\delta_2+\mu }{\gamma_2+\delta_2+\mu }\frac{\mathcal{R}_2-1}{\mathcal{R}_2},\quad R^{EE,2}_2=\frac{\gamma_2}{\delta_2+\mu }I^{EE,2}_2,\\&I^{EE,2}_1=J^{EE,2}_1=J_2^{EE,2}=R_1^{EE,2}=0.
\end{split}
\]

To study the stability of~$\phi_{EE}^1$, we compute the invasion number~$\hat{\mathcal{R}}^1_2$ of strain two at the equilibrium~$\phi_{EE}^1$
of strain one.  To do so, we compute the maximal eigenvalue of the next-generation matrix of strain two: Similar to the computation in~\ref{app:R0}, we linearize the relevant parts in the system of infected compartments,~\eqref{eq:NGM}, about~$\phi^{EE,1}$ to obtain the next-generation matrix and then compute its' maximal eigenvalue.  The resulting invasion number is given by~\eqref{eq:IEE1unstable}, see Theorem~\ref{thm:EE1}(ii).

Similarly, the invasion number~$\hat{\mathcal{R}}^2_1$ of strain one at the equilibrium~$\phi_{EE}^2$
of strain two is given by
\[
\hat{\mathcal{R}}^2_1:=\left[1+\frac{\sigma_{21}\gamma_2\eta_1}{\gamma_2+\delta_2+\mu }(\mathcal{R}_2-1)\right]\frac{\mathcal{R}_1}{\mathcal{R}_2}.
\]
For sufficiently large~$\mathcal{R}_1$,~$\phi_{EE}^2$ is unstable since~$\hat{\mathcal{R}}^2_1>1$.  Thus proving part iii.

\section{Proof of Theorem 2.2}\label{app:TheoremCE}
The following Lemma reduces the computation of the equilibrium point~$\phi^*$ from a solution to a nonlinear system of seven equations, determined by the steady-state of~\eqref{eq:model}, to a solution of a nonlinear system of three equations for three unknown~$(S^*, I_1^*, I_2^*)$:
\begin{lemma}\label{lem:reduction2Systemof3}
Let~$\sigma_{12}>0$, and let~$\phi^*=(S^*,I_1^*,I_2^*,J_1^*,J_2^*,R_1^*,R_2^*)$ be a steady-state solution of the system~\eqref{eq:model} with~$I_1^*>0$ and~$I_2^*>0$.  Then,~$\{S^*,I_1^*,I_2^*\}$ satisfy 
\begin{subequations}\label{eq:eq4CE} 
\begin{align}
        &\mu -\sum_{i=1}^{2}\beta_i[I_i^*+\eta_iJ_i^*(S^*,I_1^*,I_2^*)]S^*-\mu  S^*+\sum_{k=1}^{3}\delta_k  R_k^*(S^*,I_1^*,I_2^*)=0,\label{eq:eq4Sreduced}\\[1ex]
        &\beta_1S^*[I_1^*+\eta_{1}J_1^*(S^*,I_2^*)]-(\mu +\gamma_1)I_1^*=0,\label{eq:eq4I1reduced}\\[1ex]
         & \beta_2S^*[I_2^*+\eta_{2}J_2^*(S^*,I_1^*)]-(\mu +\gamma_2)I_2^*=0,\label{eq:eq4I2reduced}
        \end{align}
        where
                \begin{equation}\label{eq:Ji_func_S_Ii}
	J_1^*(S^*,I_2^*) =
 \begin{cases}\frac{\gamma_2}{\gamma_1+\mu } I_2^* -\frac{\mu +\delta_2}{\gamma_1+\mu }R_2^*(S^*)&\sigma_{21}>0\\
 0&\sigma_{21}=0
 \end{cases},
 \quad J_2^*(S^*,I_1^*) = \frac{\gamma_1}{\gamma_2+\mu }I_1^* -\frac{\mu +\delta_1}{\gamma_2+\mu }R_1^*(S^*).
	\end{equation}
 and
        \begin{equation}\label{eq:Ri_func_S}
R^*_1(S^*)=\frac{1-\mathcal{R}_2S^*}{\eta_2\mathcal{R}_2\sigma_{12}},\qquad R_2^*(S^*) =\begin{cases}\frac{1-\mathcal{R}_1S^*}{\eta_1\mathcal{R}_1\sigma_{21}} ,&\sigma_{21}>0\\
\frac{\gamma_2}{\delta_2+\mu}I_2^*,&\sigma_{21}=0,
\end{cases}.
\end{equation}
\end{subequations}
\end{lemma}
\begin{proof}      
When~$I_1^*>0$ and~$I_2^*>0$, Equations~\eqref{eq:eqI1} and~\eqref{eq:eqI2} imply that the steady-state solution satisfies
\[
\frac{\eta_1J_1^*}{I_1^*+\eta_1J_1^*}=1-\mathcal{R}_1S^* ,\quad \frac{\eta_2J^*_2}{I^*_2+\eta_2J^*_2}=1-\mathcal{R}_2S^*.
\]
When~$\sigma_{21}>0$, substituting the above expressions in~\eqref{eq:eqJ1} and~\eqref{eq:eqJ2} defines~$R^*_i$ as a linear function of~$S^*$
\[
\eta_1\mathcal{R}_1\sigma_{21}R_2^*(S^*) =\frac{\eta_1J_1^* }{I_1^* +\eta_1 J^*_1}=1-\mathcal{R}_1S^* ,\quad \eta_2\mathcal{R}_2\sigma_{12}R^*_1(S^*)=\frac{\eta_2J_2^* }{I_2^* +\eta_2 J_2^* }=1-\mathcal{R}_2S^*,
\]
yielding~\eqref{eq:Ri_func_S} for~$\sigma_{21}>0$.
Furthermore, the steady-state solution satisfies~$J_i^\prime(t)=R_i^\prime(t)=0$ for~$i=1,2$, hence~$J_1^\prime(t)+R_2^\prime(t)=0=J_2^\prime(t)+R_1^\prime(t)$, implying~\eqref{eq:Ji_func_S_Ii}.
Substituting~\eqref{eq:Ri_func_S} and~\eqref{eq:Ji_func_S_Ii} in~(\ref{eq:S}-\ref{eq:eqI2}), and setting~$S^\prime(t)=I_1^\prime(t)=I_2^\prime(t)=0$ gives rise to~(\ref{eq:eq4Sreduced}-\ref{eq:eq4I2reduced}).

When~$\sigma_{21}=0$, the above computation remains valid for~$J_2^*$ and~$R_1^*$, while equations~\eqref{eq:eqJ1} and~\eqref{eq:eqR2} directly imply~$J_1^*=0$ and~$(\delta_2+\mu)R_2^*=\gamma_2 I_2^*$.  Therefore, giving rise to~\eqref{eq:Ji_func_S_Ii} and~\eqref{eq:Ri_func_S} for~$\sigma_{21}=0$.
\end{proof}

We now rely on Lemma~\ref{lem:reduction2Systemof3} to prove the first part of Theorem~\ref{thm:phi*stability}.  For brevity, we consider separately the cases~$\sigma_{21}>0$ and~$\sigma_{21}=0$:
\begin{prop}[Existence, uniqueness and approximation of~$\phi^*$ for~$\sigma_{21}>0$]\label{prop:approxCEsigma21positive}
Let~$\sigma_{21}>0$.  Under the conditions of Theorem~\ref{thm:phi*stability}, if~$\sigma_{12}>0$, then there exists a unique steady-state solution $$\phi^*=(S^*,I_1^*,I_2^*,J_1^*,J_2^*,R_1^*,R_2^*)$$ of the system~\eqref{eq:model} with~$I_1^*>0$ and~$I_2^*>0$.  This solution satisfies~\eqref{eq:approxCE}.
\end{prop}
\begin{proof}
By Lemma~\ref{lem:reduction2Systemof3},~$\phi^*$ satisfies~\eqref{eq:eq4CE}.
Denote~$\varepsilon=1/\mathcal{R}_1$ and consider an expansion of the form
\begin{equation}\label{eq:series_S_I}
S^*=s_0+\varepsilon s_1+\varepsilon^2 s_2+\mathcal{O}(\varepsilon^3),\quad 
I^*_i=a_{i}+\varepsilon b_{i}+\mathcal{O}(\varepsilon^2).
\end{equation}
Substituting~\eqref{eq:series_S_I} into~\eqref{eq:eq4Sreduced} and equating the~$\mathcal{O}(1/\varepsilon)$ terms in the arising expansion yields 
\[
 [((\gamma_1 + \mu )a_1 + \eta_1\gamma_2a_2)\sigma_{21} + (\mu  + \delta_2)s_0]s_0=0,
\]
Since~$a_1,a_2,s_0\ge0$, and due to~\eqref{eq:condParms}, it follows that~$s_0=0$.
Further substituting~\eqref{eq:series_S_I} with~$s_0=0$  in~\eqref{eq:eq4I2reduced}, yields 
\[
(\gamma_2+\mu)a_2=\mathcal{O}(\varepsilon),
\]
hence~$a_2=0$.  Similarly, substituting~\eqref{eq:series_S_I}  in~\eqref{eq:eq4I1reduced} and in~\eqref{eq:eq4I2reduced} yields
\[
(\gamma_1+\mu)(s_1-1)a_1=\mathcal{O}(\varepsilon), \quad \left[a_1\sigma_{12}\mathcal{R}_2\eta_2\gamma_1 -(\delta_1+\mu)\right]s_1-(\gamma_2+\mu)\sigma_{12}b_2=\mathcal{O}(\varepsilon),
\]
respectively.  Since~$s_1,a_1,b_2\ge0$, the only feasible solution of the above system is
\[
s_1=1,\quad a_1=\frac{(\gamma_2+\mu)\sigma_{12}b_2+(\delta_1+\mu)}{\mathcal{R}_2\eta_2\gamma_1}.
\]
At the next order, the expansion of~\eqref{eq:eq4I1reduced} yields
\[
s_2=-\frac{\gamma_2\eta_1 b_2}{(\gamma_1+\mu)a_1}.
\]
Finally, substituting~\eqref{eq:series_S_I} into~\eqref{eq:eq4Sreduced} yields~\eqref{eq:s2b1_sigma21_positive}.

Requiring~$I_2^*>0$ implies~$b_2>0$, giving rise to condition~\eqref{eq:beta2cond},
see~\eqref{eq:b2a1}.  This condition also ensures that, for sufficiently large~$\beta_1>\gamma_1$,~$I_1^*>0$ since~$a_1>0$ when~$b_2>0$, see~\eqref{eq:b2a1}.

The above computation uniquely determined the values of the coefficients in~\eqref{eq:approxCE}.  We now use the implicit function theorem to show that there exists a unique solution to~\eqref{eq:eq4Sreduced}, which can be expressed as an expansion of the form~\eqref{eq:series_S_I} up to arbitrary order:  
Let~${\bf F}(\tilde S^*(\varepsilon),I_1^*(\varepsilon),I_2^*(\varepsilon);\varepsilon)$ be the algebraic system~\eqref{eq:eq4CE} where~$S^*=\frac{\gamma_1+\mu}{\beta_1}(1+\tilde S^*)$.
Then, as follows from the analysis above, ${\bf F}(\tilde S^*,I_1^*,I_2^*;\varepsilon=0)$ has a unique solution~$\tilde S^*(0)=0$, $I_1^*(0)=a_1$, $I_2^*(0)=0$.  The Jacobian matrix~${\bf J}$ of~\eqref{eq:eq4Sreduced} at~$\varepsilon=0$ equals
\[
\left.{\bf J}\right|_{(\tilde S^*(0),I_1^*(0),I_2^*(0))}=\left[\begin{array}{ccc}
-a_1(\gamma_1+\mu)&-(\gamma_1+\mu)-\delta_3\frac{\gamma_1+\gamma_2+\mu}{\gamma_2+\mu}&-\eta_1\gamma_2-\delta_3\frac{\gamma_1+\gamma_2+\mu}{\gamma_2+\mu}\\
a_1(\gamma_1+\mu)&0&\eta_1\gamma_2\\
0&0&-(\gamma_2+\mu)
\end{array}\right].
\]
The determinant of~${\bf J}$ equals
\[
\mbox{Det}({\bf J})=-(\gamma_1+\mu)(\mu^2 + (\delta_3 + \gamma_1 + \gamma_2)\mu + (\delta_3 + \gamma_2)\gamma_1 + \gamma_2\delta_3)a_1.
\]
Since~$a_1>0$ and due to~\eqref{eq:condParms},~$\mbox{Det}({\bf J})<0$, and hence the Jacobian matrix is invertible.  Therefore, by the implicit function theorem, there exists a constant~$\varepsilon_0$ and unique functions~$\tilde S^*(\varepsilon),I_1^*(\varepsilon),I_2^*(\varepsilon)$ defined for~$0<\varepsilon<\varepsilon_0$ such that
${\bf F}(\tilde S^*(\varepsilon),I_1^*(\varepsilon),I_2^*(\varepsilon);\varepsilon)=0$.  Furthermore, these functions can be represented by a series expansion of the form~\eqref{eq:series_S_I}.

\end{proof}
Next, we continue to focus on the case~$\sigma_{21}>0$ and prove that~$\phi^*$ is stable,
\begin{prop}[Stability of~$\phi^*$ for~$\sigma_{21}>0$]\label{prop:CEstable_sigma21positive}
Let~$\sigma_{21}>0$, and let~$\phi^*=(S^*,I_1^*,I_2^*,J_1^*,J_2^*,R_1^*,R_2^*)$ be the steady-state solution of the system~\eqref{eq:model} with~$I_1^*>0$ and~$I_2^*>0$.  Then, for sufficiently large~$\beta_1>\gamma_1$,~$\phi^*$ is locally stable. 
\end{prop}
\begin{proof}
By Proposition~\ref{prop:approxCEsigma21positive},~$\phi^*$ satisfies~\eqref{eq:approxCE}.  Substituting approximation~\eqref{eq:approxCE} of the multi-strain endemic equilibrium point,~$\phi^*$, into the Jacobian matrix of~\eqref{eq:model} and computing the corresponding characteristic polynomial~$P(\lambda)$ yields
\begin{subequations}\label{eq:P}\begin{equation}
\varepsilon^2P(\lambda)=a_1^2\sigma_{21}P_0(\lambda)+\varepsilon P_1(\lambda)+\varepsilon^2 P_2(\lambda)+\varepsilon^3 P_3(\lambda)+\mathcal{O}(\varepsilon^4),\qquad \varepsilon=\frac1{\mathcal{R}_1},\quad 0<\varepsilon\ll1,
\end{equation}
where
\begin{equation}\begin{split}
P_0(\lambda)&=(\lambda+\gamma_1+\mu)(\lambda+\gamma_2+\mu)Q(\lambda),\quad Q(\lambda)=\lambda^3+\alpha_2\lambda^2+\alpha_1\lambda+\alpha_0,\\
P_1(\lambda)&=(\sigma_{21} + 1)a_1\lambda^6+c_{1,5}\lambda^6+\cdots +c_{1,0},\\ P_2(\lambda)&=\lambda^7 +c_{2,6}\lambda^6+\cdots +c_{2,0},\\ P_3(\lambda)&=c_{3,6}\lambda^6+\cdots +c_{3,0},\quad c_{3,6}\ne0,
\end{split}
\end{equation}
and
\begin{equation}\label{eq:Pcoeff}
\begin{split}
\alpha_2&=b_2(\mu + \gamma_2)\sigma_{12} + 2\mu + \delta_1 + \delta_3 + \gamma_1, \\
\alpha_1&=(2b_2\sigma_{12} + 1)\mu^2 + (b_2(\delta_3 + \gamma_1 + 3\gamma_2)\sigma_{12} + \delta_1 + \delta_3 + \gamma_1)\mu + b_2\gamma_2(\delta_3 + \gamma_1 + \gamma_2)\sigma_{12} + \delta_3(\delta_1 + \gamma_1)= \\
&b_2\sigma_{12}(2\mu^2+(\delta_3 + \gamma_1 + 3\gamma_2)\gamma_2+\mu(\delta_3 + \gamma_1 + \gamma_2))+\mu^2 + (\delta_1 + \delta_3 + \gamma_1)\mu + \delta_3(\delta_1 + \gamma_1)\\
\alpha_0&=b_2\sigma_{12}[\mu^2 + (\delta_3 + \gamma_1 + \gamma_2)\mu + (\delta_3 + \gamma_1)\gamma_2 + \gamma_1\delta_3)(\mu + \gamma_2)],\\
c_{j,0}&\ne0\quad j=1,2,3.
\end{split}
\end{equation}
\end{subequations}
Since~$b_2>0$, see~\eqref{eq:b2a1}, and since~$\delta_i>0$ or~$\mu>0$, it follows that the coefficients,~$\alpha_j>0$, $j=0,1,2$, are all strictly positive.

The characteristic polynomial~$P(\lambda)$ has overall seven roots.  Five of the roots~$\lambda_i$ of~$P(\lambda)$ identify, to leading order, with the roots of~$P_0(\lambda)$.  These roots satisfy~$\lambda_i=\lambda_i^{\rm approx}+\mathcal{O}(\varepsilon)$ where 
\begin{subequations}\label{eq:lambdaapprox}    
\begin{equation}
\lambda_i^{\rm approx}=-\mu-\gamma_j<0,\quad i=1,2,
\end{equation}
and~$\{\lambda_i^{\rm approx}\}_{i=3}^5$ are the roots of
\begin{equation}
Q(x)=x^3+\alpha_2x^2+\alpha_1x+\alpha_0=0,
\end{equation}
see~\eqref{eq:P}.
The coefficients of~$Q(x)$ are all strictly positive, see~\eqref{eq:Pcoeff}, and direct computation shows that~$\alpha_2\alpha1>\alpha_0$.  Thus, the Routh–Hurwitz criterion implies that all roots of~$Q(x)$ have a negative real part.

Dominant balance shows that~$P(\lambda)$ has two additional roots that satisfy
\begin{equation}
\lambda_i^{\rm approx}=\frac{x_i}{\varepsilon},\quad \frac{|\lambda_i^{\rm approx}-\lambda_i|}{|\lambda_i|}=\mathcal{O(\varepsilon)},\qquad i=6,7,
\end{equation}
\end{subequations}
where~$x_i$ are the roots of
\[
x^2+(\sigma_{21} + 1)a_1x+a_1^2
\sigma_{21}=(x+\sigma_{21}a_1)(x+a_1).
\]
Hence,
\[
x_6=-\sigma_{21}a_1<0,\quad x_7=-a_1<0.
\]
Overall, all roots have a negative real part for sufficiently large~$\beta_1>\gamma_1$. Therefore,~$\phi^*$ is linearly stable.
\end{proof}

Finally, we consider the case~$\sigma_{21}=0$.  The proofs of the results, in this case, follow closely the proofs of Propositions~\ref{prop:approxCEsigma21positive} and~\ref{prop:CEstable_sigma21positive}.
\begin{prop}[Existence, uniqueness and approximation of~$\phi^*$ for~$\sigma_{21}=0$]\label{prop:approxCEsigma21zero}
Let~$\sigma_{21}=0$.  Under the conditions of Theorem~\ref{thm:phi*stability}, if~$\sigma_{12}>0$, then there exists a unique steady-state solution $$\phi^*=(S^*,I_1^*,I_2^*,J_1^*,J_2^*,R_1^*,R_2^*)$$ of the system~\eqref{eq:model} with~$I_1^*>0$ and~$I_2^*>0$.  This solution satisfies~\eqref{eq:approxCE}.
\end{prop}
\begin{proof}
By Lemma~\ref{lem:reduction2Systemof3},~$\phi^*$ satisfies~\eqref{eq:eq4CE}.  
In the case,~$\sigma_{21}=0$, it follows that~$J_1^*=0$ see~\eqref{eq:eqJ1} and thus~$S^*=1/\mathcal{R}_1$ when~$I_1^*>0$, see~\eqref{eq:eqI1}.
Denote~$\varepsilon=1/\mathcal{R}_1$ and consider an expansion of the form
\begin{equation}\label{eq:series_S_I_sigma21_zero}
S^*=\varepsilon,\quad I^*_i=a_{i}+\varepsilon \tilde b_{i}+\mathcal{O}(\varepsilon^2),\qquad \sigma_{21}=0.
\end{equation}
Substituting~\eqref{eq:series_S_I_sigma21_zero} in~\eqref{eq:eq4I2reduced}, yields at leading order~$a_2=0$, and at the next order
\[
a_1=\frac{(\gamma_2+\mu)\sigma_{12}b_2+(\delta_1+\mu)}{\mathcal{R}_2\eta_2\gamma_1},
\]
as in the case of~$\sigma_{21}>0$.
Further substituting~\eqref{eq:series_S_I_sigma21_zero} into~\eqref{eq:eq4Sreduced} yields~\eqref{eq:s2b1_sigma21_zero}.
As in the case of~$\sigma_{21}>0$, condition~\eqref{eq:beta2cond} ensures~$I_1^*,I_2^*>0$.

The above computation uniquely determines the values of the coefficients in~\eqref{eq:approxCE}.  Following the proof of Proposition~\ref{prop:approxCEsigma21positive}, we now use the implicit function theorem to show that there exists a unique solution to~\eqref{eq:eq4Sreduced}, which can be expressed as an expansion~\eqref{eq:series_S_I_sigma21_zero} up to arbitrary order, by showing that the relevant Jacobian matrix is invertible:  
In the case~$\sigma_{21}=0$,~$S(\varepsilon)=\varepsilon$.  Therefore, 
the Jacobian matrix~${\bf J}$ of~\eqref{eq:eq4Sreduced} at~$\varepsilon=0$ reduces to an invertible $2\times2$ matrix 
\[
\left.{\bf J}\right|_{(I_1^*(0),I_2^*(0))}=\left[\begin{array}{cc}
-(\gamma_1+\mu)-\delta_3\frac{\gamma_1+\gamma_2+\mu}{\gamma_2+\mu}&\delta_3\frac{-(\gamma_1+\gamma_2+\mu)\delta_3+\gamma_2\delta_2}{\delta_2+\mu}\\
0&-(\gamma_2+\mu)
\end{array}\right].
\]
\end{proof}

\begin{prop}[Stability of~$\phi^*$ for~$\sigma_{21}>0$]
Let~$\sigma_{21}=0$, and let~$\phi^*=(S^*,I_1^*,I_2^*,J_1^*,J_2^*,R_1^*,R_2^*)$ be the steady-state solution of the system~\eqref{eq:model} with~$I_1^*>0$ and~$I_2^*>0$.  Then, for sufficiently large~$\beta_1>\gamma_1$,~$\phi^*$ is locally stable. 
\end{prop}
\begin{proof}
Substituting~\eqref{eq:approxCE} in the Jacobian matrix of~\eqref{eq:model} and computing the corresponding characteristic polynomial~$P(\lambda)$ yields
\begin{subequations}\label{eq:P_sigma21_zero}\begin{equation}
    \label{eq:P2}
\varepsilon P(\lambda)=a_1 (\gamma_1+\mu)P_0(\lambda)+\varepsilon P_1(\lambda)+\varepsilon^2 P_2(\lambda)+\mathcal{O}(\varepsilon^3),\qquad \varepsilon=\frac1{\mathcal{R}_1},\quad 0<\varepsilon\ll1,
\end{equation}
where
\begin{equation}\begin{split}
P_0(\lambda)&=(\lambda+\gamma_1+\mu)(\lambda+\gamma_2+\mu)(\lambda+\delta_2+\mu)Q(\lambda),\quad Q(\lambda)=\lambda^3+\alpha_2\lambda^2+\alpha_1\lambda+\alpha_0,\\
P_1(\lambda)&=\lambda^7+c_{1,6}\lambda^6+c_{1,5}\lambda^5+\cdots +c_{1,0},\\ P_2(\lambda)&=c_{2,6}\lambda^6+\cdots +c_{2,0},
\end{split}
\end{equation}
and
\begin{equation}\begin{split}
\alpha_1&=\gamma_1+\delta_1+\delta_3+2\mu+(\gamma_2+\mu)\sigma_{12}b_2,\\ \alpha_2&=(\gamma_2+\mu)(\gamma_1 + \gamma_2+\delta_3+2\mu)\sigma_{12}b_2 + (\gamma_1+\delta_1+\mu)(\delta_3+\mu),\\
\alpha_3&=(\gamma_2+\mu)((\gamma_1+\delta_3)\gamma_2+\gamma_1\delta_3 + (\gamma_1 + \gamma_2+\delta_3)\mu + \mu^2)\sigma_{12}b_2.
\end{split}
\end{equation}
\end{subequations}
The coefficients of~$Q(x)$ are all strictly positive, and direct computation shows that~$\alpha_2\alpha1>\alpha_0$.  Thus, the Routh–Hurwitz criterion implies that all roots of~$Q(x)$ have a negative real part.

Dominant balance and equation of orders shows that~$P(\lambda)$ has an additional root 
\begin{equation}
\lambda_7^{\rm approx}=-\frac{a_1(\gamma_1+\mu)}{\varepsilon}<0,\qquad \frac{|\lambda_7^{\rm approx}-\lambda_7|}{|\lambda_7|}=\mathcal{O(\varepsilon)}.
\end{equation}
Overall, all roots have a negative real part for sufficiently large~$\beta_1>\gamma_1$. Therefore,~$\phi^*$ is linearly stable
\end{proof}
\section{Numerical validation of Theorem 2.2}\label{app:numericalVerification}
The asymptotic approximations presented in Theorem~\ref{thm:phi*stability} are valid for sufficiently large values of~$\mathcal{R}_1$. 
In what follows we validate the asymptotic approximation~\eqref{eq:approxCE} of Theorem~\ref{thm:phi*stability}, and show it is valid already for moderate values of~$\mathcal{R}_1$.
To do so, we compute numerically~$\phi^*$ for the values of the parameters used in the example presented in  Figure~\ref{fig:example_balanced_invasion},
and then compute the error in the approximation~\eqref{eq:approxCE}.
We note that we have considered additional sets of parameters and have reached similar results (data not shown).

Figure~\ref{fig:phiStarErr} presents the error in the approximation~\eqref{eq:approxCE} for~$\phi^*$.  As expected, we observe that~$|S^*-\varepsilon-s_2\varepsilon^2|=O(\varepsilon^3)$ and~$|I_i^*-a_i-b_i\varepsilon|=O(\varepsilon^2)$.  Formally the asymptotic results are valid for sufficiently small values of~$\varepsilon$.  Yet, we observe that these results are valid also for~$\varepsilon\approx0.5$, i.e., for~$\mathcal{R}_1\approx 2$.
\begin{figure}[ht!]
\begin{center}
\includegraphics[width=0.8\textwidth]{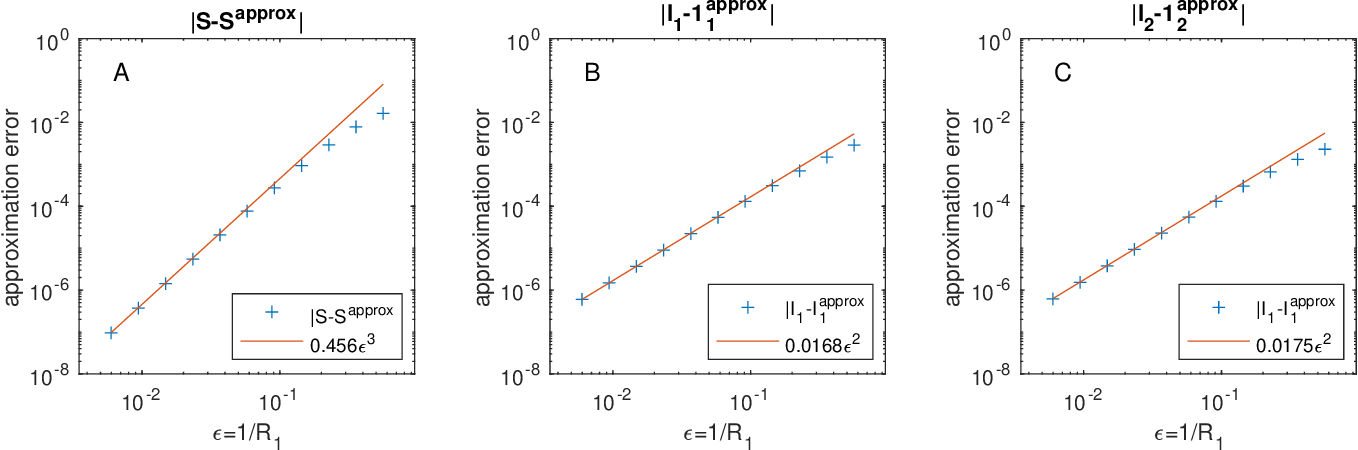}
\caption{Error in approximation~\eqref{eq:approxCE} for the parameters used in Figure~\ref{fig:example_balanced_invasion}. 
A: Error~$|S^*-\varepsilon-s_2\varepsilon^2|$ (red markers), super-imposed is the graph of~$0.456\varepsilon^3$.  B: Error~$|I_1^*-a_1-b_1\varepsilon|$.  C:  Error~$|I_2^*-b_2\varepsilon^2|$. Super-imposed in B and C are the curves proportional to~$\varepsilon^2$ that best fit the data at the left end of the graph.}
\label{fig:phiStarErr}
\end{center}
\end{figure}

\bibliographystyle{plain}

\end{document}